\documentclass[letterpaper, 10 pt, conference]{ieeeconf}  

\IEEEoverridecommandlockouts                              
\usepackage{graphics} 
\usepackage{times} 
\usepackage{amsmath} 
\usepackage{amssymb}  
\usepackage{graphicx}
\usepackage{algorithm}
\usepackage{mathtools}
\usepackage{float}
\usepackage[noend]{algpseudocode}
\let\labelindent\relax
\usepackage{enumitem}
\usepackage{color}

\newtheorem{theorem}{Theorem}

\newtheorem{lemma}{Lemma}
\newtheorem{assumption}{Assumption}

\author{Jihoon Suh \and Takashi Tanaka
\thanks{This work is supported by NSF Award 1944318. }
\thanks{Both authors are with the Department of Aerospace Engineering and Engineering Mechanics, University of Texas at Austin, TX 78712, USA.
        {\tt\small \{jihoonsuh,ttanaka\}@utexas.edu}}%
}

\begin{document}

\title{\LARGE \bf
Encrypted Value Iteration and Temporal Difference Learning over Leveled Homomorphic Encryption
}

\maketitle

\thispagestyle{empty}
\pagestyle{empty}

\begin{abstract}
We consider an architecture of confidential cloud-based control synthesis based on Homomorphic Encryption (HE). Our study is motivated by the recent surge of data-driven control such as deep reinforcement learning, whose heavy computational requirements often necessitate an outsourcing to the third party server. To achieve more flexibility than Partially Homomorphic Encryption (PHE) and less computational overhead than Fully Homomorphic Encryption (FHE), we consider a Reinforcement Learning (RL) architecture over Leveled Homomorphic Encryption (LHE).
We first show that the impact of the encryption noise under the Cheon-Kim-Kim-Song (CKKS) encryption scheme on the convergence of the model-based tabular Value Iteration (VI) can be analytically bounded. We also consider secure implementations of TD(0), SARSA(0) and Z-learning algorithms over the CKKS scheme, where we numerically demonstrate that the effects of the encryption noise on these algorithms are also minimal.
\end{abstract}

\section{INTRODUCTION}
The growing demand of applications such as smart grids \cite{8625421} and smart cities \cite{7945258} assures the important role of data usage and connectivity via advanced data-driven algorithms. 
However, the utility of advanced data-driven algorithms may be limited in many real-world control systems since  components within the network are often resource-constrained. The cloud-based control can be an appealing solution in such scenarios, although a naive outsourcing comes with a steep cost of privacy. Recent literature in control has shown that the use of HE could mitigate the issue of privacy to a certain extent. However, there are many remaining challenges in encrypted control technologies. For instance, in the current encrypted control literature, the potential of FHE is not fully utilized, even though FHE would be necessary to encrypt advanced control algorithms, such as deep reinforcement learning \cite{mnih2015human}.
RL framework has recently accomplished impressive feats with successful applications from AlphaGo Zero to traffic signal control, robot controls and many others, \cite{silver2017mastering, 5478405, 10.5555/2946645.2946684}. Its success is also bolstered by the availability of large-scale data and connectivity. Thus, we wish to study a cloud-based control synthesis architecture which can integrate two promising technologies, namely HE and RL. In particular, we examine the effects of using HE on RL both theoretically and numerically.

The first applications of HE to control systems utilized PHE since its relatively low computational overhead was suitable for real-time implementations. On the other hand, FHE or LHE can support more general classes of computations in the ciphertext domain, although its computational overhead makes it less suitable for real-time systems. In this regard, we first make an observation that the encrypted control can be better suited to the control synthesis problems rather than control implementations. For instance, explicit model predictive control (MPC) or RL problems require heavy computations or a large set of data to synthesize the control policy but implementing such policy may be kept local as they are often relatively light computations. In addition, real-time requirement is less stringent in control law synthesis problems than control law implementation counterpart.

\subsection{Related Work}
Various HE schemes were applied to the linear controller in \cite{Kogiso2015CybersecurityEO, FAROKHI2016163, KIM2016175}. Importance of quantization and using non-deterministic encryption scheme was identified in \cite{Kogiso2015CybersecurityEO}. Necessary conditions on encryption parameters for closed-loop stability were shown by \cite{FAROKHI2016163}. The feasibility of FHE in encrypted control was first shown in \cite{KIM2016175} by managing multiple controllers. Evaluation of affine control law in explicit MPC was shown by \cite{8126799}. In \cite{8619835} encrypted implicit MPC control was shown to be feasible but the cost of privacy as increased complexity was identified. Real-time proximal gradient method was used in \cite{DARUP2018535} to treat encrypted implicit MPC control and it showed undesired computation loads of encrypted control system. Privacy and performance degradation was further highlighted in \cite{8814398} through experiments on motion control systems. Dynamic controller was encrypted in \cite{kim2019comprehensive} using FHE exploiting the stability of the system while not relying on bootstrapping. Despite significant progress, encrypted control has been limited to simple computations where the motivation for cloud computing is questionable. 

\subsection{Contribution}
To study the feasibility of RL over HE, this paper makes the following contributions. First, as a first step towards more advanced problems of private RL, we formally study the convergence of encrypted tabular VI. We show that the impact of the encryption-induced noise can be made negligible if the $Q$-factor is decrypted in each iteration. Second, we present implementation results of temporal difference (TD) learning (namely, TD(0), SARSA(0), and Z-learning \cite{Todorov11478}) over the CKKS encryption scheme and compare their performances with un-encrypted cases. Although the formal performance analysis for this class of algorithms are difficult with encryption noise, we show numerically that these RL schemes can be implemented accurately over LHE.

\subsection{Preview}
In section \ref{Preliminaries}, we summarize the relevant essentials of RL and HE. In section \ref{Problem}, we set up the encrypted control synthesis problem and specialize it to the encrypted model-based and model-free RL algorithms. We analyze how the encryption-induced noise influences the convergence of standard VI. In section \ref{Enc_algorithms}, we perform the simulation studies to demonstrate the encrypted model-free RL over HE. Finally, in section \ref{conclusion}, we summarize our contributions and discuss the future research directions.

\section{Preliminariess}
\label{Preliminaries}
\subsection{Reinforcement Learning}
RL can be formalized through the finite state Markov decision process (MDP). We define a finite MDP as a tuple $(\mathcal{S}, \mathcal{A}, P, R)$, where $\mathcal{S} = (1, 2, \dots, n)$ is a finite state space, $\mathcal{A}$ is a finite action space, $P = P(s_{t+1}|s_t, a_t)$ is a Markov state transition probability and $R = R(s_t, a_t, s_{t+1})$ is a reward function for the state-transition. A policy can be formalized via a sequence of stochastic kernels $\pi=(\pi_0, \pi_1, ...)$, where $\pi_t(a_t|s_{0:t+1}, a_{0:t})$ is a mapping for each state $s_t \in\mathcal{S}$ to the probability of selecting the action $a_t\in\mathcal{A}$ given the history $(s_{0:t-1}, a_{0:t-1})$. 

We consider a discounted problem with a discount factor $\gamma \in [0, 1)$ throughout this paper. Under a stationary policy $\pi$, the Bellman's optimality equation can be defined through a recursive operator $T$ applied on the initial value vector $V^{\pi}_0$ of the form $V^{\pi} = (V(1), \dots,  V(n))^{\pi}$:
\begin{equation}
    \begin{aligned}[b]
        (TV^{\pi})(s_t) = \max_{a_t} \sum_{s_{t+1}}P[R + \gamma V^\pi(s_{t+1})].
    \end{aligned}
    \label{eq:bellman_eq}
\end{equation}

Throughout this paper, we adopt the conventional definition for the maximum norm $\|\cdot\|_{\infty}$. The following results are standard \cite{Bertsekas2009}.
\begin{lemma}
\label{lemma:optimal value}
    \begin{enumerate}[label=(\alph*)]
    \item For any vectors $V$ and $\bar{V}$, we have $$\|TV - T\bar{V}\|_{\infty} \leq \gamma\|V - \bar{V}\|_{\infty}.$$
    \item The optimal value vector $V^*$ is the unique solution to the equation $V^* = TV^*$,
    \item We have $$\lim_{k\to\infty} T^k V = V^*$$ for any vector $V$.
    \end{enumerate}
\end{lemma}
If some policy $\pi^*$ attains $V^*$, we call $\pi^*$ the optimal policy and the goal is to attain the optimal policy $\pi^*$ for the given MDP environment.

Often the state values are written conveniently as a function of state-action pairs called $Q$-values:
\begin{equation}
    \begin{aligned}[b]
        Q^{\pi}(s_t, a_t) = \sum_{s_{t+1}}P[R + \gamma V^\pi(s_{t+1})],
    \end{aligned}
    \label{eq:Q-values}
\end{equation}
and for all state $s_t$, $V^{\pi}_{k+1}(s_t) = \max_{a_t \in \mathcal{A}} Q^{\pi}(s_t, a_t)$.

RL is a class of algorithms that solves MDPs when the model of the environment (e.g., $P$ and/or $R$) is unknown. RL can be further classified into two groups: model-based RL and model-free RL \cite{silver}. A simple model-based RL first estimates the transition probability and reward function empirically (to build an artificial model) then utilizes the VI to solve the MDP. Model-free RL on the other hand directly computes the value function by averaging over episodes (Monte Carlo methods) or by estimating the values iteratively like TD learning. 

TD(0) is one of the simplest TD learning algorithm where the tuple $(s, R, s')$ is sampled from the one-step ahead observation \cite{sutton_barto_2018}. The on-policy iterative update rule for estimating the value is called TD(0) and is written as follows:
\begin{equation}
    \hat{V}_{t+1}^\pi(s) = \hat{V}_{t}^\pi(s) + \alpha_t(s) \delta_t.
    \label{eq:TD(0)_update_rule}
\end{equation}
$\delta_t$ denotes the TD error and is computed with the sample as
\[
    \delta_t = R(s,a) + \gamma \hat{V}_{t}^\pi(s') - \hat{V}_{t}^\pi(s),
\]
starting with an initial guess $\hat{V}^{\pi}_0$. We use $\hat{V}$ to indicate that the values are estimated. With some standard assumptions on the step size $\alpha_t(s)$ and exploring policies, the convergence of TD(0) is standard in the literature, see \cite{325119}.

\subsection{Homomorphic Encryption}
HE is a structure-preserving mapping between the plain-text domain $\mathcal{P}$ and the cipher-text domain $\mathcal{C}$. Thus, encrypted data in $\mathcal{C}$ with a function $f(\mathcal{C})$ can be outsourced to the cloud for confidential computations.

PHE supports only the addition or the multiplication. On the other hand, LHE can support both the addition and the multiplication but the noise growth of ciphertext multiplication is significant without the bootstrapping \cite{homenc}. The bootstrapping operation can promote a LHE scheme to FHE enabling unlimited number of multiplications but it requires a heavy computation resource on its own. 

In order to use both addition and multiplication necessary in evaluating RL algorithms, we use an LHE. In particular, we employ CKKS encryption scheme proposed in \cite{cryptoeprint:2016:421} as it is well suited for engineering applications and also supports parallel computations. Also, there exists a widely available library optimized for implementations. We will only list and describe here several key operations and properties with more detailed information found in Appendix. We are particularly concerned with the fact that the CKKS encryption is noisy (due to error injection for security) but if properly implemented, the noise is manageable.

The security of CKKS scheme assumes the difficulty of the Ring Learning With Errors (RLWE) problem. Security parameters are a power-of-two integer $\mathcal{N}$, a ciphertext modulus $\mathcal{Q}$ and the variance $\sigma$ used for drawing error from a distribution. The operation \textit{KeyGen} uses security parameters ($\mathcal{N}$, $\mathcal{Q}$, $\mathcal{\sigma}$) to create a $\lambda$-bit public key \textit{pk}, a secret key \textit{sk}, and an evaluation key \textit{evk}.

\section{Encrypted learning over the cloud}
\label{Problem}
\subsection{Cloud-based reinforcement learning}
Our control loop consists of the plant (environment) and the controller (client) and the cloud server. We are concerned with a possible data breach at the cloud. In order to prevent the privacy compromise of our data such as training data set and other parameters at play, we add the homomorphic encryption module as seen in Fig. \ref{fig:2}. In particular, we ignore the malleability risks \cite{inproceedings} inherent to the considered HE scheme, which is beyond the scope of this paper. In our context, $Syn(Enc(\cdot))$ can be thought of as a ciphertext domain implementation of RL algorithms. However, $Syn(Enc(\cdot))$ is a general control synthesis that can be evaluated remotely. Ideal uses for $Syn(Enc(\cdot))$ would include advanced data-driven control synthesis procedures such as MPC, or Deep RL. We emphasize that the encryption adds communication delays for control policy synthesis but does not add an extra computation cost to control implementation. For RL, the controller's task is to implement the most up-to-date state-action map $\pi$ that is recently synthesized, which can be done locally.

We assume the tabular representation of the state and action pairs. The client explores its plant and samples the information set $h_t = (s_{t+1} , a_{t} , r_{t})$. Other necessary parameters such as learning rate $\alpha_{t}(s)$ and $\gamma$ are grouped as $\theta$ denoting the hyperparmeters. This set of data is then encrypted by the CKKS encryption module to generate ciphertexts $\textbf{c}_v$, $\textbf{c}_{v'}$, $\textbf{c}_\alpha$, $\textbf{c}_\gamma$, and $\textbf{c}_r$, which are encrypted data for $\hat{V}_{t}^\pi(s)$, $\hat{V}_{t}^\pi(s')$, $\alpha_t(s)$, $\gamma$, and $R(s,a)$, and will be used to implement \eqref{eq:TD(0)_update_rule}. Note that the subscripts refer to the value when these ciphertexts are decrypted. The cloud is instructed on how to evaluate the requested algorithms in ciphertexts. One may concern about the cloud knowing the form of the algorithm, but this can be resolved by the circuit privacy. That is, we can hide the actual computation steps by adding zeros or multiplying ones in ciphertexts on the original algorithm. We also assume that the cloud is semi-honest so that the cloud faithfully performs the algorithm as instructed. The output of the cloud is a newly synthesized policy $\pi_t(s_{0:t})$, which, after a decryption, can be accessed on real-time by the client to produce the control action $a_t(s_t)$.

\begin{figure}
    \centering
    \includegraphics[scale=0.25]{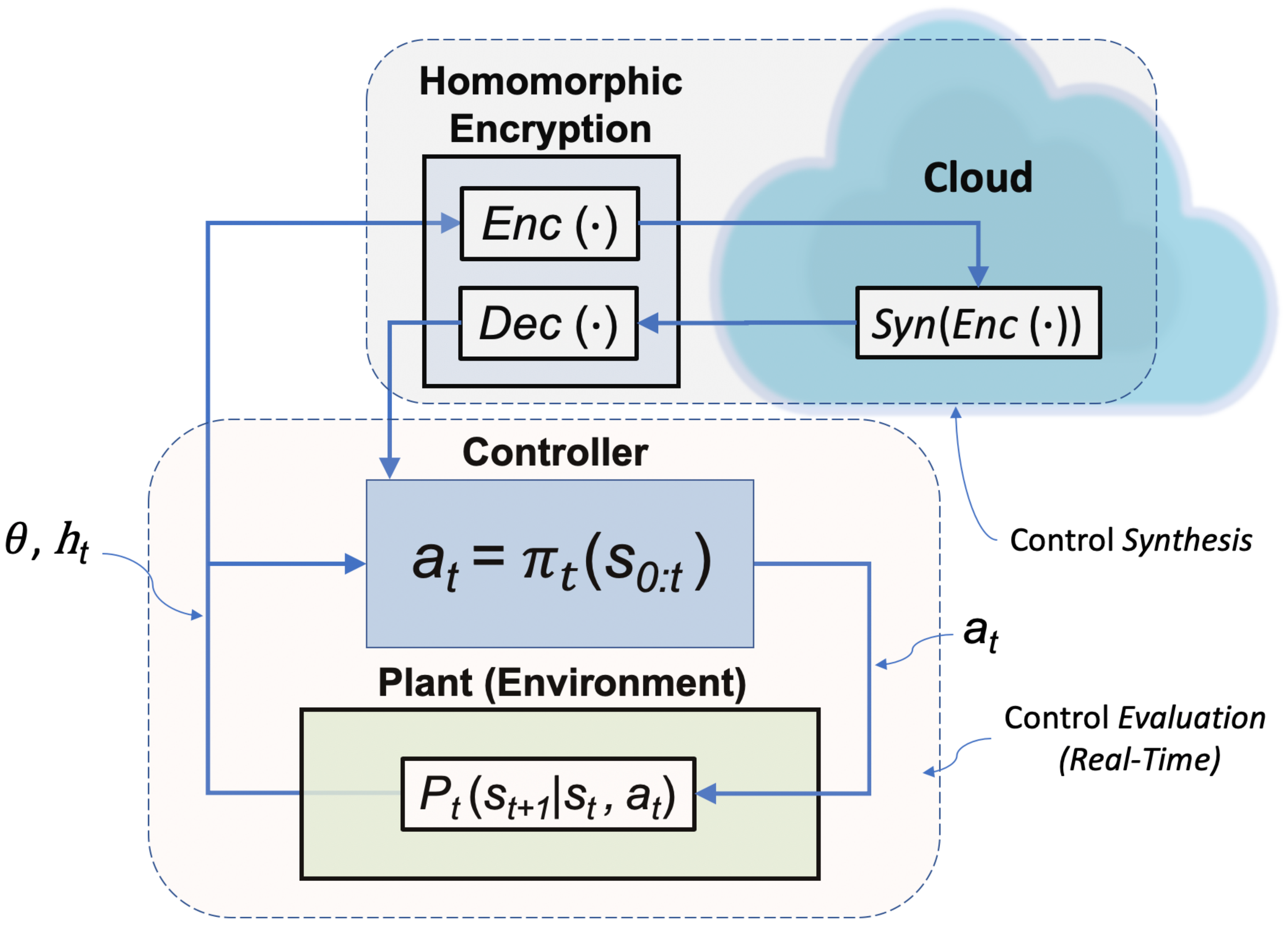}
    \caption{Encrypted RL Control synthesis with cloud in the loop.}
    \label{fig:2}
\end{figure}

As a first step towards investigating more sophisticated RL algorithms, we specialize the proposed framework to solving RL problems in finite MDP with more elementary methods. In particular, we consider the basic model-based RL and model-free RL using tabular algorithms. For the model-based, we theoretically analyze the convergence of VI under the presence of encryption-induced noise. For the model-free, we implement the encrypted TD algorithms to estimate the values and investigate how the encryption noise affects the output.

\subsection{Encrypted Model-based RL}
The client is assumed to explore the plant, sample the information sets $h_t$ and build the model first. A simple model of the state-transition probability can be in the form
\begin{equation}
    P(s_{t+1}|s_t, a_t) = \frac{N(s, a, s')}{N(s, a)},
    \label{eq:model}
\end{equation}
where $N(s, a, s')$ denotes the number of visits to the triplet $(s, a, s')$ and $N(s, a)$ to the pair $(s, a)$ with $s, s' \in  \mathcal{S}, a \in \mathcal{A}$. Similarly, the reward function $R(s_t, a_t, s_{t+1})$ can be quantified by the average rewards accumulated per the state-action pair.

Then, the client can encrypt the initial values of the state along  with the model $P$ and $R$ and the discount rate $\gamma$ and request the cloud to perform the computation of \eqref{eq:Q-values} over the ciphertext domain. Since comparison over HE is non-trivial, the max operation is difficult to be implemented homomorphically. Thus, the cloud computes the $Q$ values and the controller will receive the decrypted $Q$ values and completes the VI process for iteration index $k$:

\begin{equation}
    \tilde{V}^{\pi}_{k+1}(s_t) = \max_{a_t \in \mathcal{A}} \tilde{Q}^{\pi}(s_t, a_t).
\end{equation}

However, note that the $Q$ values computed by the cloud will be noisy due to the encryption, hence we use $\tilde{Q}$ and $\tilde{V}$ to denote the noisy value. Let $w(s_t, a_t)$ be the encryption-induced noise produced by computations over HE. Then, the noisy $Q$ values can be written as
\begin{equation}
    \tilde{Q}^{\pi}(s_t, a_t) = Q^{\pi}(s_t, a_t) + w(s_t, a_t),
    \label{eq:noisy Q}
\end{equation}
where the noise term is bounded such that $|w(s_t, a_t)| \leq \epsilon$ $\forall s_t, a_t$ for some $\epsilon > 0$. Appendix $A$ shows how $\epsilon$ depends on the encryption parameter and operations used to evaluate the given algorithm. 

We now analyze the discrepancy between two sequences of vectors $V^{\pi}_k$ and $\tilde{V}^{\pi}_k$. We separate the analysis into the synchronous and the asynchronous cases. 

\subsubsection{Synchronous VI} The synchronous VI means that the VI is applied to all state $s$ simultaneously. First note that $V^{\pi}_k$ is computed by the noiseless VI
\begin{equation}
    \begin{aligned}[b]
    V^{\pi}_{k+1}(s_t) & =  \max_{a_t \in \mathcal{A}} \sum_{s_{t+1}}P[R + \gamma V^\pi_t(s_{t+1})] \\
                 & =  (TV^{\pi}_k)(s_t) \quad \forall s_t \in \mathcal{S},
     \end{aligned}
     \label{eq:noiselesss VI}
\end{equation}
and $\tilde{V}^{\pi}_k$ is computed by the noisy VI
\begin{equation}
    \begin{aligned}[b]
    \tilde{V}^{\pi}_{k+1}(s_t) & =  \max_{a_t \in \mathcal{A}} \sum_{s_{t+1}}P[R + \gamma V^\pi_k(s_{t+1}) + w(s_t, a_t)]
     \end{aligned}
     \label{eq:noisy VI}
\end{equation}
We will utilize the following simple lemma, whose proof is straightforward and hence omitted.
\begin{lemma}
\label{lemma:max}
For any arbitrary vectors $x=(x(1),\ldots,x(n))$ and $w=(w(1), \dots, w(n))$ such that $\|w\|_{\infty} \leq \epsilon$,
\begin{equation}
    \max_{i} (x(i) + w(i)) = \max_{i}(x(i)) + \tilde{w}
\end{equation}
for some constant $\tilde{w}$ satisfying $|\tilde{w}| \leq \epsilon$.
\end{lemma}
By applying Lemma $2$ on the noisy VI \eqref{eq:noisy VI},\\
we obtain
\begin{equation}
    \label{eq:noisy VI2}
    \begin{aligned}[b]
        \tilde{V}^{\pi}_{k+1}(s_t) & =  \max_{a_t \in \mathcal{A}} \sum_{s_{t+1}}P[R + \gamma \tilde{V}^\pi_k(s_{t+1}) + w_k(s_t, a_t)] \\
        & = \max_{a_t \in \mathcal{A}} \sum_{s_{t+1}}P[R + \gamma \tilde{V}^\pi_k(s_{t+1})] + \tilde{w}_k(s_t) \\
        & = (T\tilde{V}^\pi_k)(s_t) + \tilde{w}_k(s_t),
     \end{aligned}
\end{equation}
where the vector $\tilde{w}_k$ satisfies $\|\tilde{w}_k\|_{\infty} \leq \epsilon$. In the vector form, $\eqref{eq:noisy VI2}$ can be written as
\begin{equation}
    \label{eq:noisy VI2-2}
    \begin{aligned}[b]
    \tilde{V}^{\pi}_{k+1} & = T\tilde{V}^{\pi}_k + \tilde{w}_k, \\
                          & = \tilde{T}\tilde{V}^{\pi}_k.
    \end{aligned}
\end{equation}
We can now quantify the worst-case performance degradation due to the encryption-induced noise. The result is summarized in the next Theorem.

\begin{theorem}(Approximate VI, \cite{Bertsekas2009})
\label{thm:conv_synch}
Let $V^*$ be the optimal value function characterized by Lemma \eqref{lemma:optimal value} ($a$). Suppose $\tilde{V}^{\pi}_k$ is the sequence of vectors computed by the noisy VI \eqref{eq:noisy VI}. For an arbitrary initial condition $\tilde{V}_0$, we have
$$
\limsup\limits_{k\rightarrow\infty} \|V^* - \tilde{V}^{\pi}_k\|_{\infty} \leq \frac{\epsilon}{1-\gamma}.
$$
\end{theorem}
\begin{proof}
Let $V^\pi_k$ be the sequence of vectors computed by the noiseless VI \eqref{eq:noiselesss VI} with arbitrary initial conditions $V^\pi_0$ and $\tilde{V}^{\pi}_0$. Then,

\begin{subequations}
\label{eq:pf1}
    \begin{align}
        \|V^{\pi}_{k+1} - \tilde{V}^{\pi}_{k+1}\|_{\infty} &= \|TV^{\pi}_{k} - T\tilde{V}^{\pi}_{k+1} - \tilde{w}_k\|_{\infty} \\ \label{eqn:pf1-1}
                                & \leq  \|TV^{\pi}_{k} - T\tilde{V}^{\pi}_{k}\|_{\infty} + \|\tilde{w}_k\|_{\infty} \\ \label{eqn:pf1-2}
                                & \leq \|TV^{\pi}_{k} - T\tilde{V}^{\pi}_{k}\|_{\infty} + \epsilon  \\ 
                                & \leq \gamma\|V^{\pi}_{k} - \tilde{V}^{\pi}_{k}\|_{\infty} + \epsilon,
    \end{align}
\end{subequations}
where the first inequality \eqref{eqn:pf1-1} follows from \eqref{eq:noisy VI2-2}, and \eqref{eqn:pf1-2} follows from the triangular inequality, and the last inequality is due to Lemma \eqref{lemma:optimal value}. Now define a sequence  $e_k$ of positive numbers by
\begin{equation}
    \label{eq:err}
    e_{k+1} = \gamma e_k + \epsilon
\end{equation}
with $e_0 = \|V^{\pi}_0 - \tilde{V}^{\pi}_0\|_{\infty}$.
By geometric series,
\[
    \limsup\limits_{k\rightarrow\infty} e_k = \frac{\epsilon}{1-\gamma}.
\]
Comparing \eqref{eq:pf1} and \eqref{eq:err}, we have by induction
\[
\|V^{\pi}_k - \tilde{V}^{\pi}_k\|_{\infty} \leq e_k, \forall k = 0, 1, 2, \dots.
\]
Therefore,
\begin{equation}
    \label{eq:pf1_end}
    \limsup\limits_{k\rightarrow\infty} \|V^{\pi}_k - \tilde{V}^{\pi}_k\|_{\infty} \leq \frac{\epsilon}{1-\gamma}.
\end{equation}
Since \eqref{eq:pf1_end} holds for any $V^{\pi}_0$, we can pick $V^{\pi}_0 = V^*$ for which we have $V^{\pi}_k = T^kV^* = V^*$ for $k = 0, 1, 2, \dots$ by lemma \eqref{lemma:optimal value}, we obtain the desired result.
\end{proof}

\subsubsection{Asynchronous VI}
It is often necessary or beneficial to run the VI algorithm asynchronously, or state-by--state, for simulations. We can define the asynchronous noiseless and noisy VI with the new mapping $F$ and $\tilde{F}$, respectively.
\begin{equation}
     FV^{\pi}_{k}(s_t) = 
     \begin{cases}
        (TV^{\pi}_{k})(s) & \text{if } s = s_t,\\
        V^{\pi}_k(s_t),                                         & \text{otherwise},
\end{cases}
\label{eq:async}
\end{equation}

\begin{equation}
\label{eq:asynch_noisy}
     \tilde{F}\tilde{V}^{\pi}_{k}(s_t) = 
     \begin{cases}
        (\tilde{T}\tilde{V}^{\pi}_{k})(s) & \text{if } s = s_t,\\
        \tilde{V}^{\pi}_k(s_t),                                         & \text{otherwise}.
\end{cases}
\end{equation}
In each case, we make the following assumption. 
\begin{assumption}
    \label{assume:A1}
    \begin{enumerate}[label=(\alph*)]
    \item Each state is visited for updates infinitely often.
    \item There exists a finite constant $M$, which is greater than or equal to the number of updates to sweep through each state at least once.
    \end{enumerate}
\end{assumption}

Assumption 1 is essential for the following theorem as it ensures that the mapping $F$ and $\tilde{F}$ are contraction operators as long as all the states are visited at least once and that the time it takes for visiting all the states at least once is finite. A common approach to ensure this assumption would be to incorporate exploring actions as seen in $\epsilon$-greedy policy.

Now, define the iteration sub-sequence $\{k_n\}_{n=0,1,2,\dots}$ such that $k_0 = 0$ and each state is visited at least once between $k_{n+1}$ and $k_n$. For instance, consider a finite MDP with $4$ states denoted $s_1,s_2,s_3,s_4$. If the state trajectory is $(s_1, s_2, s_4, s_3, s_1, s_1, s_2, s_4, s_3, s_1, s_1, ...)$ for $t = (1, 2, 3, 4, 5, 6, 7, 8, 9, 10, 11, ...)$, then the sub-sequence $k_n$ formed is $k_0 = 0, k_1 = 4, k_2 = 9, \dots$.
\begin{theorem}
\label{thm:conv_asynch}
Let $V^*$ be the optimal value function as defined previously. Suppose $\tilde{V}^{\pi}_{k_n}$ is the sequence of vectors computed by the asynchronous noisy VI \eqref{eq:asynch_noisy} under Assumption 1. For an arbitrary initial condition $\tilde{V}_0$, we have

$$
\limsup\limits_{n\rightarrow\infty} \|V^* - \tilde{V}^{\pi}_{k_n}\|_{\infty} \leq \frac{M\epsilon}{1-\gamma}.
$$

\end{theorem}
\begin{proof}
By definition of mapping $F$ and $\tilde{F}$, we can write:
    \begin{equation}
        \begin{aligned}[b]
        \|V^{\pi}_{k_{n+1}} - \tilde{V}^{\pi}_{k_{n+1}}\|_{\infty} & = \|F^{k_{n+1} - k_n}V^{\pi}_{k_n} - \tilde{F}^{k_{n+1} - k_n}\tilde{V}^{\pi}_{k_n}\|_{\infty} \\
         & \leq \|F^{k_{n+1} - k_n}V^{\pi}_{k_n} - \tilde{F}^{k_{n+1} - k_n}\tilde{V}^{\pi}_{k_n}\|_{\infty} \\
           & \quad \quad + (k_{n+1} - k_n)\epsilon
        \end{aligned}
    \end{equation}
Similar to the proof of Theorem $1$, the inequality is due to the bound of the error vector and the triangle inequality. Now, we note that the asynchronous mapping is a contraction mapping with respect  to a sequence index $k_n$. Thus, we have
\[
\|V^{\pi}_{k_{n+1}} - \tilde{V}^{\pi}_{k_{n+1}}\|_{\infty} \leq \gamma\|V^{\pi}_{k_n} - \tilde{V}^{\pi}_{k_n}\|_{\infty} + (k_{n+1} - k_n)\epsilon.
\]
By forming a sequence with $n$ as it is done in the previous proof, and by the Assumption $1-$(b), it yields
\begin{equation}
    \label{eq:asynch_pf2}
     \limsup\limits_{n\rightarrow\infty} \|V^{\pi}_{k_n} - \tilde{V}^{\pi}_{k_n}\|_{\infty} \leq \frac{M\epsilon}{1-\gamma}.
\end{equation}
By expanding the left hand side and applying the reverse traingle inequality
\begin{equation}
    \begin{aligned}[b]
    \label{eq:asynch_pf2-2}
    \|V^{\pi}_{k_n} - \tilde{V}^{\pi}_{k_n}\|_{\infty} & =  \|V^{\pi}_{k_n} -  V^* - (\tilde{V}^{\pi}_{k_n} - V^*)\|_{\infty} \\
    & \geq \|\tilde{V}^{\pi}_{k_n} - V^*\|_{\infty} - \|V^{\pi}_{k_n} - V^*\|_{\infty},
    \end{aligned}
\end{equation}
for which we know the second term approaches zero in the limit. Since the left hand side is bounded in the limit with constants, we achieve the proposed result.
\end{proof}

Theorem \ref{thm:conv_asynch} assures that the encrypted VI outsourced to the cloud also guarantees a comparable performance asymptotically.

\subsection{Encrypted Model-free RL}
Consider again the TD(0) update rule, \eqref{eq:TD(0)_update_rule}. The cloud receives the set of ciphertexts $\{\textbf{c}_v$, $\textbf{c}_{v'}$, $\textbf{c}_\alpha$, $\textbf{c}_\gamma$, $\textbf{c}_r\}$. Then, the update rule in the ciphertext domain becomes:
\begin{equation}
    \textbf{c}_{v}(t+1) = \textbf{c}_{v}(t) + \textbf{c}_{\alpha}\cdot\textbf{c}_{r} + \textbf{c}_{\alpha}\cdot\textbf{c}_{\gamma}\cdot\textbf{c}_{v'}(t) - \textbf{c}_{\alpha}\cdot\textbf{c}_{v}(t).
    \label{eq:HE_TD(0)_update_rule}
\end{equation}

Upon decryption of $\textbf{c}_{v}(t+1)$, we need to remember that the computed value is corrupted with the noise. A similar error analysis for the TD algorithms can be performed (perhaps using stochastic approximation theory). However, the formal analysis in this domain presents some difficulties. Whereas a conventional theory on stochastic approximation with an exogenous noise seen in \cite{Bertsekas2009} requires the noise to approach zero in the limit and bounded, the encryption-induced noise satisfies only the bounded condition. We thus only provide some implementation results at this time to gain some insight. We hope to rigorously prove this case as in the future work.

\section{Simulation}
We present implementation results of various model-free (TD(0), SARSA(0), and Z-learning) RL algorithms over the CKKS implementation scheme. The environment used is the grid world with the state size $|\mathcal{S}| = 36$ for all three and the action size $|\mathcal{A}| = 9$ for the first two. The available actions are up, up-right, right, down-right, down, down-left, left, up-left, and stay. The reward is fixed as randomly set real numbers to simulate unknown environment. The client starts with a policy $\pi$ at the box coordinate (1,1), top-left and the grid world has three trap states and one goal state, marked by letters \textbf{T} and \textbf{G}, which terminate the current episode. In each episode, we set the maximum number of steps at which the current episode terminates as well. The learning parameter set $\theta$ consists of the discount factor, learning rate, and exploration percentage. As soon as the new data set $h_t$ containing the reward (or cost) and values of states $s$ and $s'$ become available, the client uploads the encrypted data to the cloud.

Choosing encryption parameters is not straightforward but there exists a standardization effort, \cite{HomomorphicEncryptionSecurityStandard}. Also, an open-source library SEAL \cite{sealcrypto} provides a practical tutorial and accessible tools along with encryption parameter generator. We use the default 128-bit security ($\lambda = 128$) encryption parameters ($\mathcal{N}, \mathcal{Q}, \sigma$) generated by SEAL with the user input of $\mathcal{N}$. These are listed in Table \ref{tab:sec_param}. The size of $\mathcal{N}$ need not be this large as there is no parallel operation exploited for the particular example application considered in this paper. However, future applications such as multi agents RL or deep RL can find such capability useful as they contain many batch operations.

The CKKS encryption without employing a bootstrapping allows a predetermined depth for multiplication. Thus, for interested users, it is important to note the largest depth of ciphertext multiplications needed to evaluate the algorithm at hand. For example, TD(0) update rule considered in equation \eqref{eq:HE_TD(0)_update_rule} requires $\textbf{c}_{\alpha}\cdot\textbf{c}_{\gamma}\cdot\textbf{c}_{v'}(t)$ at the most. This is factored into the design your encryption parameters.

To examine the effect of encryption noise, we created two tables. One keeps track of the values of un-encrypted updates and the other keeps track of the values updated over HE. We recorded the error between two values through each iteration. At final stages of learning, these errors were confirmed to be bounded by some constant of very small magnitude. Although formal convergence analysis of model-free algorithms such as TD(0) are currently not available in this paper, simulation results suggest that they can be performed in the encrypted domain as equally well. Formal analysis of these algorithms based on the analysis already done is left as future research.
\begin{table}
    \begin{center}
    \caption{Encryption Parameters}
    \label{tab:sec_param}
    \begin{tabular}{|c|c|c|c|}
        \textbf{Param.} & \textbf{TD(0)} & \textbf{SARSA(0)} & \textbf{Z-learning} \\ 
        \hline
        \textbf{\textbf{$\mathcal{N}$}} & 8192 & 8192 & 16384 \\
        \textbf{$\mathcal{Q}$} & 219 & 219 & 441 \\
        \textbf{$\sigma$} & $\frac{8}{\sqrt{2\pi}}$ & $\frac{8}{\sqrt{2\pi}}$ & $\frac{8}{\sqrt{2\pi}}$ \\
        \end{tabular}
    \end{center}
\end{table}

\subsection{Prediction: TD(0)}
We implement a GLIE (greedy in the limit with infinite exploration) type learning policy seen in \cite{singh2000convergence}, where the client starts completely exploratory ($\varepsilon = 1.00$) and slowly becomes greedy ($\varepsilon = 0.00$) with more episodes. The learning rate is set to be $0$ for non-visited states and for visited states, we set $\alpha(s,t) = \frac{500}{500 + n(s,t)}$, with $n(s,t)$ counting the number of visits to the state $s$ at time $t$, to satisfy the standard learning rate assumptions. The discount factor is $\gamma = 0.9$. \textbf{*RULE*} for TD(0) is the right hand side of equation \eqref{eq:HE_TD(0)_update_rule}.

\label{Enc_algorithms}
\begin{algorithm}
\label{algo:encrypted_TDLearn}
\caption{Encrypted TD Learning}
\textbf{Client (\emph{Start})}
\begin{algorithmic}[1]
\State Perform an action $a$ and state transition $s \rightarrow s'$ on a policy $\pi_t$ to earn a reward $r$.
\State Collect the transition data set $h_t$. Index values from the current table using $h_t$.
\State Encode values (Q-values if SARSA(0); Z-values if Z-learning) and the set $\theta$ into the encoded messages $V_t$ and $\Theta$, respectively.
\State Encrypt $V_t$ and $\Theta$ to get $\tilde{V}$ and $\tilde{\Theta}$ and upload to the cloud.
\end{algorithmic}

\textbf{Cloud}
\begin{algorithmic}[1]
\State Extract ciphertexts from $\tilde{V}$ and $\tilde{\Theta}$. 

\textit{Example}: extract $\{\textbf{c}_z$, $\textbf{c}_{z'}$, $\textbf{c}_\alpha$, $\textbf{c}_l$, $\textbf{c}_{\frac{l}{2}}\}$ (Z learning).
\State Update: use the \textbf{*RULE*} with  ciphertexts extracted.
\State Upload the result of \textbf{*RULE*}, denoted by ciphertect \textbf{c*}, back to the Client.
\end{algorithmic}

\textbf{Client}
\begin{algorithmic}[1]
\State Decrypt the ciphertext \textbf{c*} to get the updated $\tilde{H}$.
\State Decode $\tilde{H}$ to get the newly synthesized table of values.
\State Update the policy $\pi_t$ if necessary.
\State go to \emph{\textbf{Start}}
\end{algorithmic}
\end{algorithm}

\subsection{Control: SARSA(0)}
The update rule for SARSA(0) is:
\begin{equation}
    \begin{split}
        \hat{Q}_{t+1}(s, a) = & \hat{Q}_{t}(s, a) + \alpha_t(s, a)\delta_t^{S},
        \label{eq:sarsa}
    \end{split}
\end{equation}
where $\delta_t^{S} = \left(r(s,a) + \gamma \hat{Q}_{t}(s', a') - \hat{Q}_{t}(s, a)\right).$

The policy for SARSA(0) is also a GLIE (greedy in the limit with infinite exploration) type learning policy used in TD(0). The learning rate $\alpha(s,t)$ and the discount factor $\gamma$ are unchanged. \textbf{*RULE*} for SARSA(0) is the right hand side of equation \eqref{eq:HE_TD(0)_update_rule} after substituting $\textbf{c}_{v}$ and $\textbf{c}_{v'}$ with $\textbf{c}_{Q}$ and $\textbf{c}_{Q'}$.

\begin{figure}[H]
  \centerline{\includegraphics[scale=0.36]{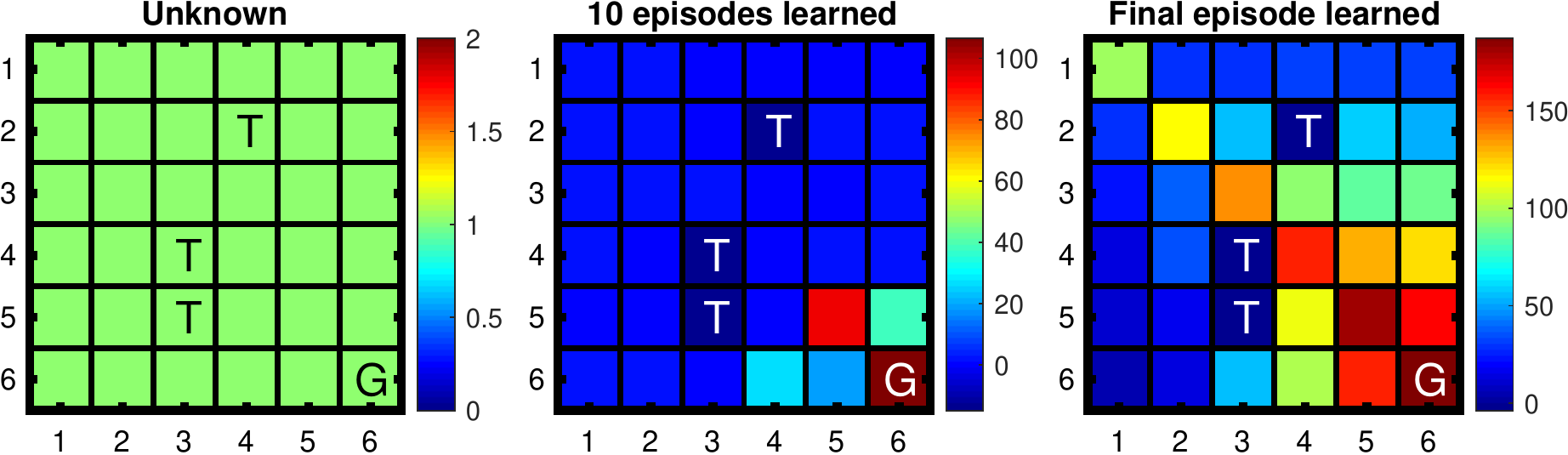}}
  \label{fig:HETD0_value_map}
  \caption{State-value map learned over encrypted TD(0); states are numbered from 1-36 (from left-right and top-to-bottom) with G (goal state), and T (trap states) marked.}
\end{figure}
\vspace{-3ex}
\begin{figure}[H]
  \centerline{\includegraphics[scale=0.36]{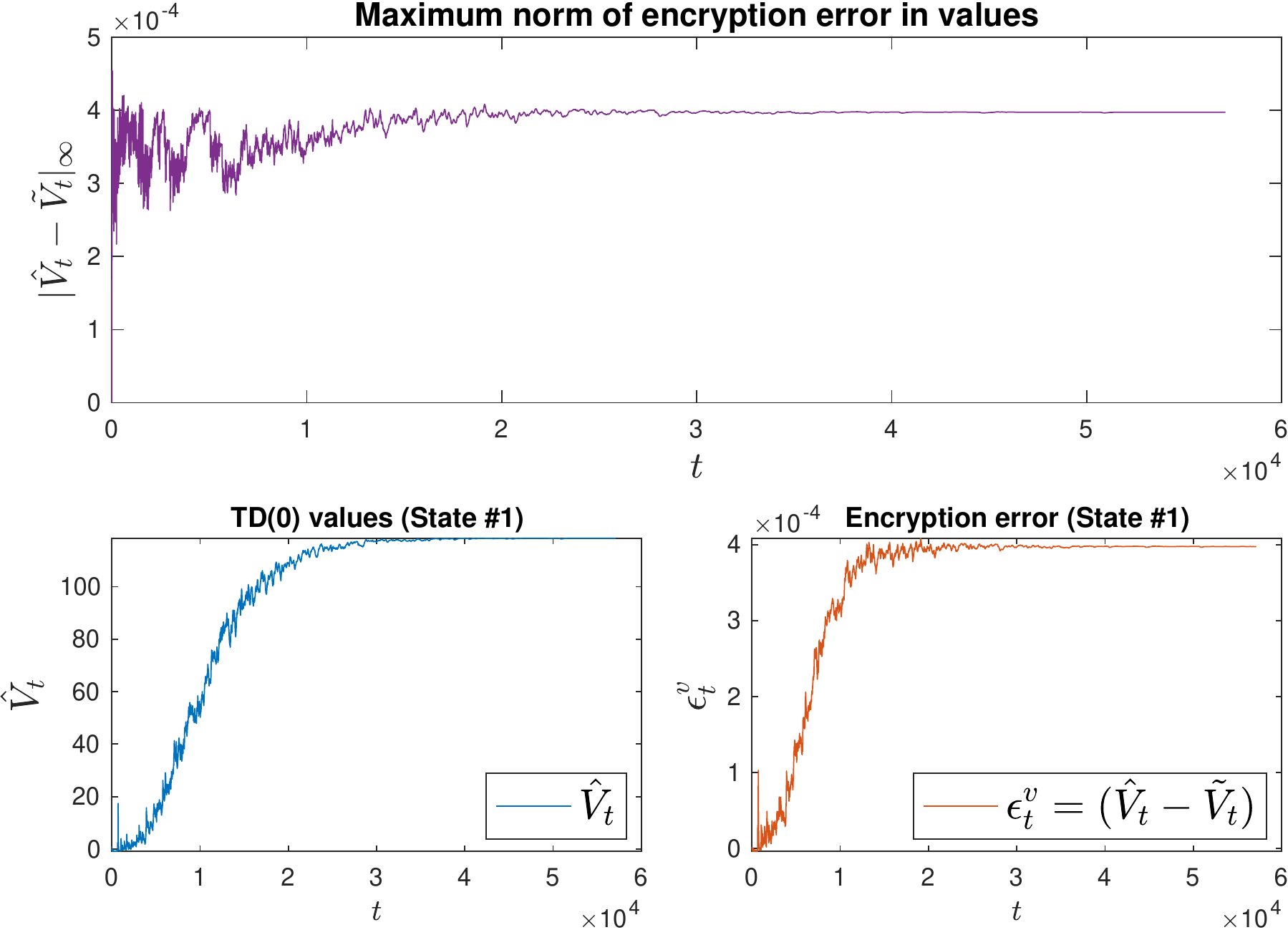}}
  \label{fig:HETD0_maxnorm}
  \caption{Maximum norm of encryption errors in TD(0) over each iterate $t$ (top). As seen in the un-encrypted value updates (bottom-left) and encryption error (bottom-right) plots of state 1, the value of state 1 gets increased as it propagates from reaching goal states more often once the policy becomes greedy (see Fig. 2) and its error dominates between iterates $t = 30000$ and $35000$.}
\end{figure}

\begin{figure}[H]
  \centerline{\includegraphics[scale=0.36]{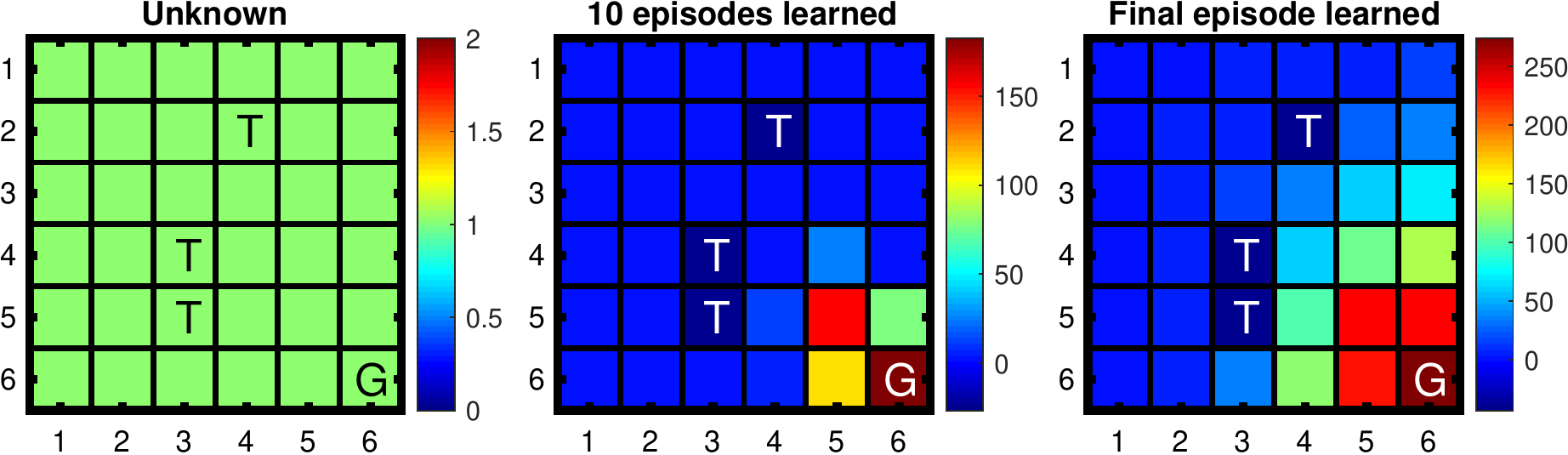}}
  \label{fig:HESARSA0_value_map}
  \caption{Value map learned over encrypted SARSA(0)}
\end{figure}
\vspace{-3ex}
\begin{figure}[H]
  \centerline{\includegraphics[scale=0.36]{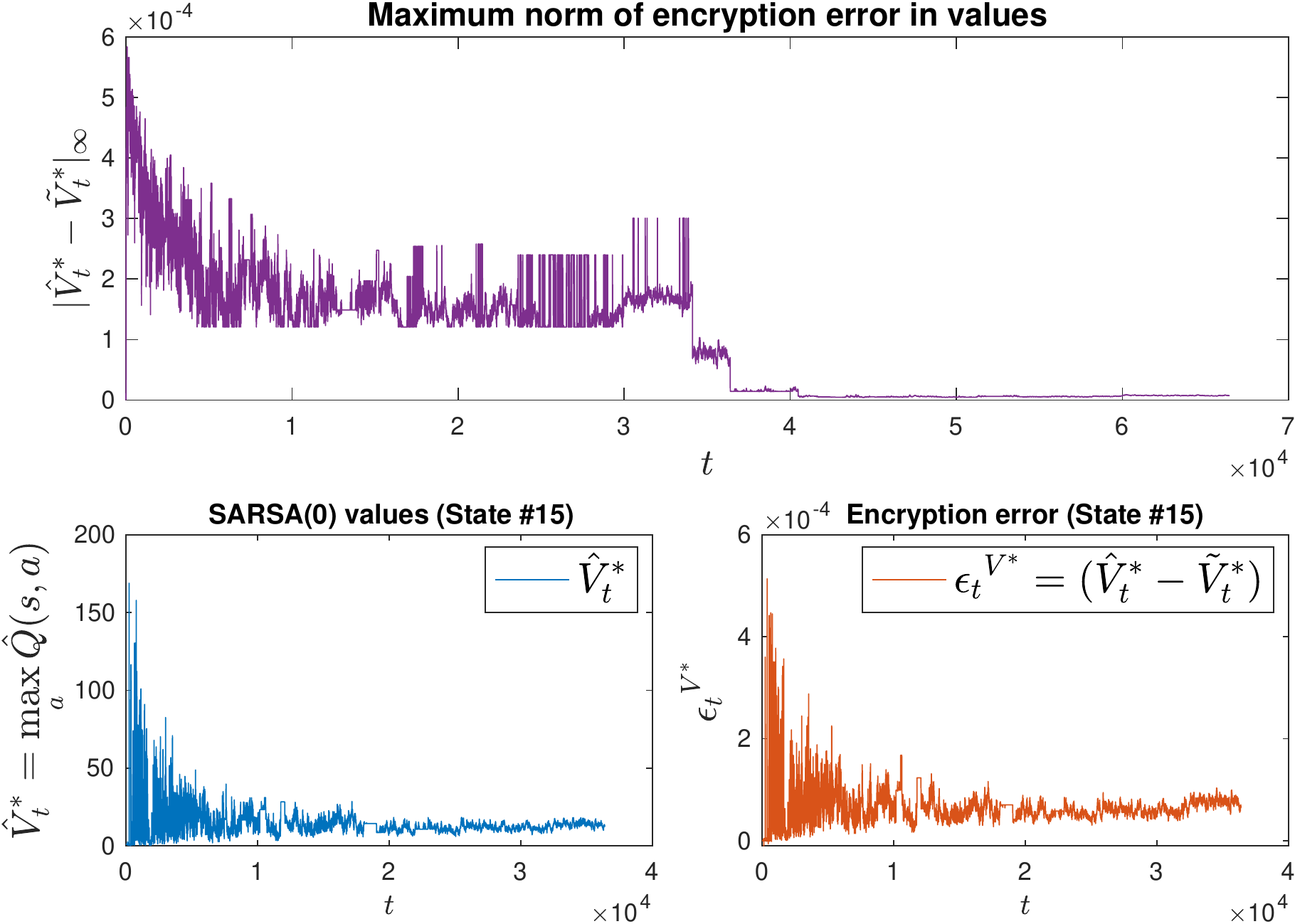}}
  \label{fig:HESARSA_maxnorm}
  \caption{Maximum norm of encryption errors in SARSA(0) over each iterate $t$ (top). The un-encrypted and decrypted values for each state are computed using the equation $V^* = \max_{a}\hat{Q}(s,a)$. Bottom-left shows the sample history of value $V^*(s = 15)$ and the associated encryption error $\hat{V}_t^\pi(s = 15) - \tilde{V}_t^\pi(s = 15)$}.
\end{figure}

\begin{figure}[H]
  \centerline{\includegraphics[scale=0.36]{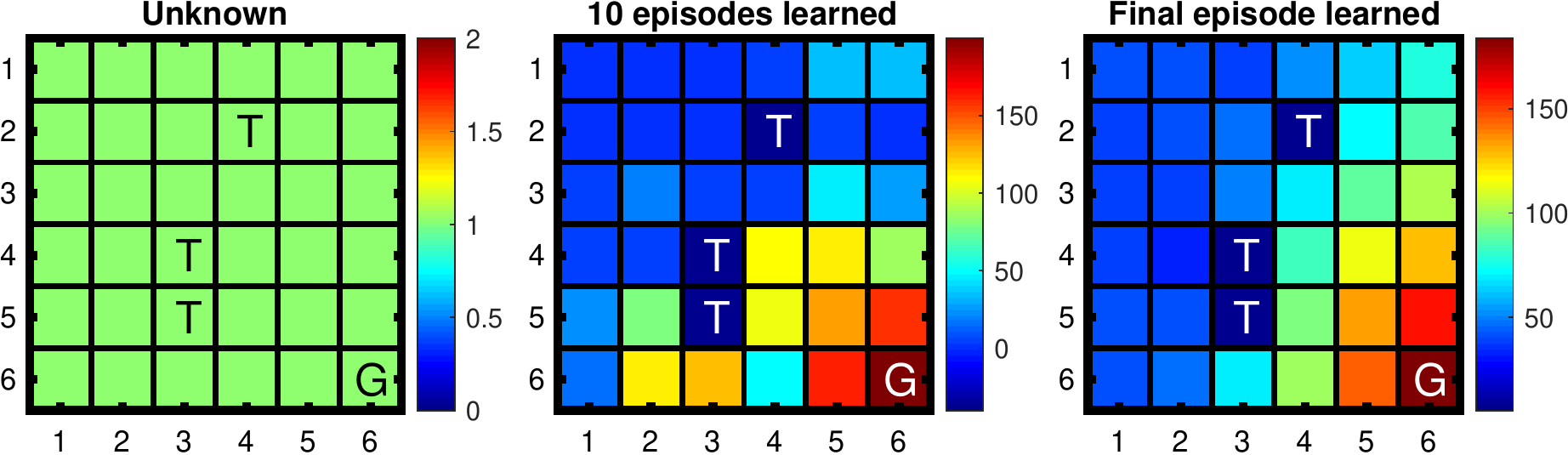}}
  \label{fig:HEZ_value_map}
  \caption{Z value map learned over encrypted Z-learning}
\end{figure}
\vspace{-3ex}
\begin{figure}[H]
  \centerline{\includegraphics[scale=0.36]{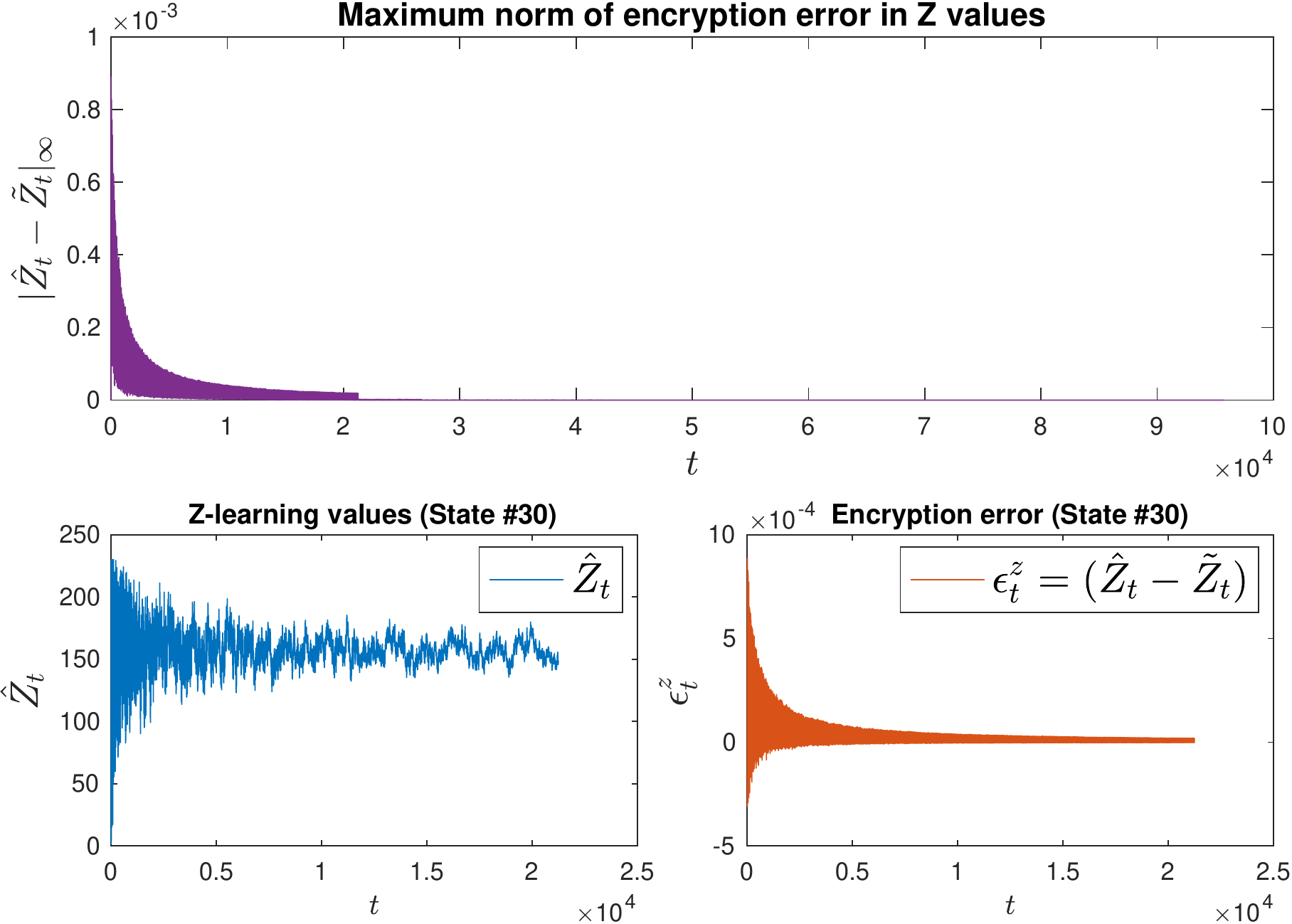}}
  \label{fig:HEZ_maxnorm}
  \caption{Maximum norm of encryption errors in Z-learning over each iterate $t$ (top). The state 30 converges to the highest Z value as seen in Fig. 4. Un-encrypted Z value updates show more variance than TD(0) counterparts due to constant exploration. Encryption error in state 30 is the highest error until about $t = 22000$ which can be seen in maximum norm but it is bounded by a very small number.}
\end{figure}

\subsection{Control: Z-Learning}
An off-policy learning control called Z-learning was proposed in \cite{Todorov11478}. It is formulated on the key observation that the control action can be regarded as an effort to change the passive state transition dynamics. The update rule for Z-learrning is:
\begin{equation}
    \hat{Z}_{t+1}(s) = \hat{Z}_t(s) + \alpha_t(s)\delta_t^{Z},
    \label{eq:Zlearning_update_rule}
\end{equation}
where $\delta_t^{Z} = exp(-l_{t})\hat{Z}_t(s') - \hat{Z}_t(s)$. The estimate $Z$ is named the desirability function and it uses the cost $l_{t}$ associated with the unknown state-transitions, rather than the reward. Z-learrning by default is exploratory. Thus, we fix the policy to be greedy ($\varepsilon = 1.00$) but it will continuously explore, too. The learning rate $\alpha(s,t)$ and the discount factor are unchanged. \textbf{*RULE*} for Z-learning is evaluating the right hand side of equation \eqref{eq:Zlearning_update_rule} after encrypting. Since evaluating $exp(-l_{t})$ over a ciphertext is not straightforward, we approximate with Taylor series.

\subsection{Results}
We observed that encryption-induced noise over time approached some small numbers with minimal fluctuations. Moreover, the effects of encryption-induced noise were minimal regarding the convergence of values. This observation complements the analysis on VI in Section \ref{Problem} as the TD-algorithms are sampling based value estimation algorithms.

Encryption and decryption are all done at the client's side and so significant computation time is expected for the client. For this reason, the ideal application will require less frequent needs for uploading and downloading and more advanced synthesis procedure that operate on a large set of data.

\section{Conclusion and Future Work}
\label{conclusion}
We considered an architecture of confidential cloud-based
RL over LHE. For the model-based RL, we showed that the impact of the encryption noise on the convergence performance can be analytically bounded. For the model-free RL, we numerically tested implementations of TD(0), SARSA(0), and Z-learning and numerically confirmed that the impacts of the encryption noise on these algorithms are also minimal. Although the applications considered do not necessarily require the cloud, we can develop the framework to adopt to more advanced synthesis algorithms in the future.

There are numerous directions to extend this paper. First, the effort to derive analytical performance guarantees for encrypted RL (including the model-free schemes considered in this paper) is necessary to prove the utility of the encrypted RL concept.
Second, an encrypted RL scheme that does not require periodic decryption (similar to the case in \cite{kim2019dynamic}) is highly desired as the periodic decryption and communication between the cloud and the controller is costly. Finally, more extensive numerical experiments are needed to fully understand the potential of advanced RL (e.g., deep Q learning) over HE. The interplay between computational overhead, delay, accuracy and security levels must be studied from both theoretical and experimental perspectives.

\section{Appendix}
\label{appendix}
\subsection{CKKS Encryption Scheme}
CKKS encoding procedure maps the vector of complex numbers sized $\frac{\mathcal{N}}{2}$ to the message $m$ in plain-text space $\mathcal{P}$. The plaintext $\mathcal{P}$ is defined as the set $\mathbb{Z}_\mathcal{Q}[\mathcal{X}]/\left(\mathcal{X}^\mathcal{N} + 1\right)$, where $\mathbb{Z}_\mathcal{Q}[\mathcal{X}]$ denotes the polynomials of $x \in \mathcal{X}$ whose coefficients are integers modulo $\mathcal{Q}$, and $\mathcal{X}^\mathcal{N} + 1$ denotes the degree $\mathcal{N}$ cyclotomic polynomials. This allows for CKKS encryption scheme to accept multiple complex-valued inputs of a size $\frac{\mathcal{N}}{2}$ at once, which is convenient for many computing applications.

Then, encryption on the message $m \in \mathcal{P}$ yields the ciphertext $\textbf{c} \in \mathcal{C}$
\begin{equation}
        Enc_{pk}(m) = \textbf{c},
    \label{eq:ckks_enc}
\end{equation}

Each ciphertext is designated with a level $L$, which indicates how many multiplications you can perform before decryption fails. This is a key limitation in comparison to FHE.

The encryption also creates a noise polynomial $e$. A properly encrypted ciphertext has a bound on the message $\|m\|_\infty \leq p$ and also on the noise $\|e\|_\infty \leq B$ when the norm is defined on those polynomials. Thus, a CKKS ciphertext can be defined as a tuple $\textbf{c} = (c, L, p, B)$. Then, decryption can be defined as follows.
\begin{equation}
        Dec(\textbf{c}, sk) \equiv m + e \pmod{q_{L}},
    \label{eq:ckks_dec}
\end{equation}
where $q_L$ is the coefficient modulus at level $L$.

For ciphertexts $\textbf{c}_i = Enc_{pk}(m_i)$ at the same level $L$,
\begin{equation}
    \begin{aligned}[b]
        Add(\textbf{c}_1, \textbf{c}_2) & = \textbf{c}_{add}  \\
                                        & = (c_{add}, L, p_1 + p_2, B_1 + B_2),
    \end{aligned}
    \label{eq:ckks_add_1}
\end{equation}
\begin{equation}
        Dec(Add(\textbf{c}_1, \textbf{c}_2), sk) \equiv m_1 + m_2 + e_{add} \pmod{q_L},
    \label{eq:ckks_add_2}
\end{equation}
where $\|e_{add}\|_{\infty} \leq B_1 + B_2$.

\bibliographystyle{IEEEtran}
\bibliography{IEEEabrv.bib, biblio.bib}

\begin{thebibliography}{10}
\providecommand{\url}[1]{#1}
\csname url@rmstyle\endcsname
\providecommand{\newblock}{\relax}
\providecommand{\bibinfo}[2]{#2}
\providecommand\BIBentrySTDinterwordspacing{\spaceskip=0pt\relax}
\providecommand\BIBentryALTinterwordstretchfactor{4}
\providecommand\BIBentryALTinterwordspacing{\spaceskip=\fontdimen2\font plus
\BIBentryALTinterwordstretchfactor\fontdimen3\font minus
  \fontdimen4\font\relax}
\providecommand\BIBforeignlanguage[2]{{%
\expandafter\ifx\csname l@#1\endcsname\relax
\typeout{** WARNING: IEEEtran.bst: No hyphenation pattern has been}%
\typeout{** loaded for the language `#1'. Using the pattern for}%
\typeout{** the default language instead.}%
\else
\language=\csname l@#1\endcsname
\fi
#2}}

\bibitem{8625421}
E.~{Hossain}, I.~{Khan}, F.~{Un-Noor}, S.~S. {Sikander}, and M.~S.~H. {Sunny},
  ``Application of big data and machine learning in smart grid, and associated
  security concerns: A review,'' \emph{IEEE Access}, vol.~7, pp.
  13\,960--13\,988, 2019.

\bibitem{7945258}
M.~{Mohammadi}, A.~{Al-Fuqaha}, M.~{Guizani}, and J.~{Oh}, ``Semisupervised
  deep reinforcement learning in support of \uppercase{I}o\uppercase{T} and
  smart city services,'' \emph{IEEE Internet of Things Journal}, vol.~5, no.~2,
  pp. 624--635, April 2018.

\bibitem{mnih2015human}
V.~Mnih, K.~Kavukcuoglu, D.~Silver, A.~A. Rusu, J.~Veness, M.~G. Bellemare,
  A.~Graves, M.~Riedmiller, A.~K. Fidjeland, G.~Ostrovski, \emph{et~al.},
  ``Human-level control through deep reinforcement learning,'' \emph{Nature},
  vol. 518, no. 7540, p. 529, 2015.

\bibitem{silver2017mastering}
D.~Silver, J.~Schrittwieser, K.~Simonyan, I.~Antonoglou, A.~Huang, A.~Guez,
  T.~Hubert, L.~Baker, M.~Lai, A.~Bolton, Y.~Chen, T.~Lillicrap, F.~Hui,
  L.~Sifre, G.~van~den Driessche, T.~Graepel, and D.~Hassabis, ``Mastering the
  game of go without human knowledge,'' \emph{Nature}, pp. 354--, Oct. 2017.

\bibitem{5478405}
I.~{Arel}, C.~{Liu}, T.~{Urbanik}, and A.~G. {Kohls}, ``Reinforcement
  learning-based multi-agent system for network traffic signal control,''
  \emph{IET Intelligent Transport Systems}, vol.~4, no.~2, pp. 128--135, June
  2010.

\bibitem{10.5555/2946645.2946684}
S.~Levine, C.~Finn, T.~Darrell, and P.~Abbeel, ``End-to-end training of deep
  visuomotor policies,'' \emph{J. Mach. Learn. Res.}, vol.~17, no.~1, p.
  1334–1373, Jan. 2016.

\bibitem{Kogiso2015CybersecurityEO}
K.~Kogiso and T.~Fujita, ``Cyber-security enhancement of networked control
  systems using homomorphic encryption,'' \emph{2015 54th IEEE Conference on
  Decision and Control (CDC)}, pp. 6836--6843, 2015.

\bibitem{FAROKHI2016163}
F.~Farokhi, I.~Shames, and N.~Batterham, ``Secure and private cloud-based
  control using semi-homomorphic encryption,'' \emph{IFAC-PapersOnLine},
  vol.~49, no.~22, pp. 163 -- 168, 2016.

\bibitem{KIM2016175}
J.~Kim, C.~Lee, H.~Shim, J.~H. Cheon, A.~Kim, M.~Kim, and Y.~Song, ``Encrypting
  controller using fully homomorphic encryption for security of cyber-physical
  systems,'' \emph{IFAC-PapersOnLine}, vol.~49, no.~22, pp. 175 -- 180, 2016.

\bibitem{8126799}
M.~{Schulze Darup}, A.~{Redder}, I.~{Shames}, F.~{Farokhi}, and D.~{Quevedo},
  ``Towards encrypted {MPC} for linear constrained systems,'' \emph{IEEE
  Control Systems Letters}, vol.~2, no.~2, pp. 195--200, 2018.

\bibitem{8619835}
A.~B. {Alexandru}, M.~{Morari}, and G.~J. {Pappas}, ``Cloud-based {MPC} with
  encrypted data,'' in \emph{2018 IEEE Conference on Decision and Control
  (CDC)}, 2018, pp. 5014--5019.

\bibitem{DARUP2018535}
M.~S. Darup, A.~Redder, and D.~E. Quevedo, ``Encrypted cloud-based {MPC} for
  linear systems with input constraints,'' \emph{IFAC-PapersOnLine}, vol.~51,
  no.~20, pp. 535 -- 542, 2018.

\bibitem{8814398}
K.~{Teranishi}, M.~{Kusaka}, N.~{Shimada}, J.~{Ueda}, and K.~{Kogiso}, ``Secure
  observer-based motion control based on controller encryption,'' in \emph{2019
  American Control Conference (ACC)}, 2019, pp. 2978--2983.

\bibitem{kim2019comprehensive}
J.~Kim, H.~Shim, and K.~Han, ``Comprehensive introduction to fully homomorphic
  encryption for dynamic feedback controller via lwe-based cryptosystem,''
  \emph{CoRR}, vol. abs/1904.08025, 2019.

\bibitem{Todorov11478}
E.~Todorov, ``Efficient computation of optimal actions,'' \emph{Proceedings of
  the National Academy of Sciences}, vol. 106, no.~28, pp. 11\,478--11\,483,
  2009.

\bibitem{Bertsekas2009}
D.~P. Bertsekas, \emph{Neuro-dynamic programmingNeuro-Dynamic
  Programming}.\hskip 1em plus 0.5em minus 0.4em\relax Boston, MA: Springer US,
  2009.

\bibitem{silver}
\BIBentryALTinterwordspacing
D.~Silver, ``Lecture 8: Integrating learning and planning.'' [Online].
  Available: \url{https://www.davidsilver.uk/teaching/}
\BIBentrySTDinterwordspacing

\bibitem{sutton_barto_2018}
R.~S. Sutton and A.~G. Barto, \emph{Reinforcement learning: an
  introduction}.\hskip 1em plus 0.5em minus 0.4em\relax The MIT Press, 2018.

\bibitem{325119}
J.~N. {Tsitsiklis}, ``Asynchronous stochastic approximation and
  {$Q$-learning},'' in \emph{Proceedings of 32nd IEEE Conference on Decision
  and Control}, Dec 1993, pp. 395--400 vol.1.

\bibitem{homenc}
C.~Gentry, ``A fully homomorphic encryption scheme,'' Ph.D. dissertation,
  Stanford University, 2009, \url{crypto.stanford.edu/craig}.

\bibitem{cryptoeprint:2016:421}
J.~H. Cheon, A.~Kim, M.~Kim, and Y.~Song, ``Homomorphic encryption for
  arithmetic of approximate numbers,'' Cryptology ePrint Archive, Report
  2016/421, 2016, \url{https://eprint.iacr.org/2016/421}.

\bibitem{inproceedings}
K.~Teranishi and K.~Kogiso, ``Control-theoretic approach to malleability
  cancellation by attacked signal normalization,'' 09 2019.

\bibitem{HomomorphicEncryptionSecurityStandard}
M.~Albrecht, M.~Chase, H.~Chen, J.~Ding, S.~Goldwasser, S.~Gorbunov, S.~Halevi,
  J.~Hoffstein, K.~Laine, K.~Lauter, S.~Lokam, D.~Micciancio, D.~Moody,
  T.~Morrison, A.~Sahai, and V.~Vaikuntanathan, ``Homomorphic encryption
  security standard,'' HomomorphicEncryption.org, Toronto, Canada, Tech. Rep.,
  November 2018.

\bibitem{sealcrypto}
``{M}icrosoft {SEAL} (release 3.4),'' \url{https://github.com/Microsoft/SEAL},
  Oct. 2019, microsoft Research, Redmond, WA.

\bibitem{singh2000convergence}
S.~Singh, T.~Jaakkola, M.~L. Littman, and C.~Szepesv{\'a}ri, ``Convergence
  results for single-step on-policy reinforcement-learning algorithms,''
  \emph{Machine learning}, vol.~38, no.~3, pp. 287--308, 2000.

\bibitem{kim2019dynamic}
J.~Kim, H.~Shim, and K.~Han, ``Dynamic controller that operates over
  homomorphically encrypted data for infinite time horizon,'' \emph{arXiv
  preprint arXiv:1912.07362}, 2019.

\end{thebibliography}
\end{document}